\documentclass[journal]{IEEEtran}

\ifCLASSINFOpdf
\else
\fi

\usepackage{cite}
\usepackage[pdftex]{graphicx}
\setkeys{Gin}{clip=true,draft=false}
\DeclareGraphicsExtensions{.pdf}
\usepackage[cmex10]{amsmath}
\usepackage{amsthm}
\usepackage{amsfonts}
\usepackage[cspex,bbgreekl]{mathbbol}
\usepackage[hang]{subfigure}
\usepackage{color}
\newtheorem{theorem}{Theorem}

\renewcommand{\baselinestretch}{0.99}

\usepackage[cmex10]{amsmath}
\usepackage{amssymb, amsthm}
\usepackage{mathtools}
\usepackage{textcomp}

\makeatother

\begin{document}
\title{Memcomputing Numerical Inversion with Self-Organizing Logic Gates}

\author{Haik Manukian, Fabio L. Traversa, Massimiliano Di Ventra\thanks{The authors are with the Department of Physics, University of California-San Diego, 9500  Gilman Drive, La Jolla, California 92093-0319, USA, e-mail: hmanukia@ucsd.edu, fabio.traversa@polito.it, diventra@physics.ucsd.edu}}

\markboth{}%
{Shell \MakeLowercase{\textit{et al.}}: Bare Demo of IEEEtran.cls for Journals}
\maketitle

\begin{abstract}
We propose to use Digital Memcomputing Machines (DMMs), implemented with self-organizing logic gates (SOLGs), to solve the problem of numerical inversion. Starting from fixed-point scalar inversion 
 we describe the generalization to solving linear systems and matrix inversion. This method, when realized in hardware, will output the result in only one computational step. As an example, we perform simulations of the scalar case using a 5-bit logic circuit made of SOLGs, and show that the circuit successfully performs the inversion. Our method can be extended efficiently to any level of precision, since we prove that producing $n$-bit precision in the output requires extending the circuit by at most $n$ bits. This type of numerical inversion can be implemented by DMM units in hardware, it is scalable, and thus of great benefit to any real-time computing application. 
\end{abstract}


% \fi
\IEEEpeerreviewmaketitle
\section{Introduction}
\IEEEPARstart{I}{n} recent decades, a growing interest into novel approaches to computing has been brewing, leading to several suggestions such as quantum computing, liquid-state machines, neuromorphic computing, etc.\cite{rese} \cite{neuromorphic} \cite{nielsen2010quantum}. Along these lines, a new computational paradigm has been recently introduced by two of us (FLT and MD), based on the 
mathematical concept of Universal Memcomputing Machines (UMMs) \cite{umm}. This novel approach utilizes memory to both store and process information (hence the name ``memcomputing'' \cite{thepa}), and in doing so it has been shown to solve complex problems efficiently \cite{memnp}. The fundamental difference with Turing-like machines as implemented with current architectures, i.e., von Neumann, is that the latter do not allow for an instruction fetch and an operation on data to occur at the same time, since these operations are performed at distinct physical 
locations that share the same bus. Memcomputing machines, on the other hand, can circumvent this problem by incorporating the program instructions not in some physically separate memory, but encoded within the physical {\it topology} of the machine. 

UMMs can be defined and built as fully analog machines \cite{memnp}. But these are plagued by requiring increasing precision depending on the problem size for measuring input/output values. Therefore, they suffer from noise issues and hence have limited scalability. Alternatively, a subset of UMMs can be defined as {\it digital} machines~\cite{dmm}. Digital Memcomputing Machines (DMMs) map integers into integers, and therefore, like our present 
digital machines are robust against noise and easily scalable. A practical 
realization of such machines has been suggested in Ref.~\cite{dmm}, where a new set of logic gates, named self-organizing logic gates (SOLGs), has been introduced. Such logic gates, being non-local in time, can {\it adapt} to signals incoming 
from {\it any} terminal. Unlike the standard logic gates, the proposed SOLGs are non-sequential, namely they can satisfy their logic proposition irrespective of 
the directionality in which the signal originates from: input-output literals or output-input literals. When these gates are then assembled in a circuit with other SOLGs (and possibly other 
standard circuit elements) that represents -- via its topology -- a particular Boolean problem, they will {\it collectively} attempt to satisfy all logic gates at once (intrinsic parallelism). After a transient that rids of the ``logical defects'' in the circuit via an instantonic phase~\cite{topo}, and which scales at most polynomially with the input size, the system then converges to the solution of the 
represented Boolean problem~\cite{dmm}.

In this paper, we suggest to apply these digital machines within their SOLG realization to numerical scalar inversion. We will also briefly mention 
how to generalize this approach to linear systems and matrix inversion. These problems scale polynomially with input size, so one would expect  
our present work to be of limited benefit. On the contrary, numerical scalar and matrix inversion, and solving linear systems are serious bottlenecks in a wide variety of science and engineering applications 
such as real-time computing, optimization, computer vision, communication, deep learning, and many others \cite{compvis} \cite{numalg} \cite{nnets} \cite{extremelearn} \cite{wirelessinv}. An efficient hardware implementation of such numerical operations would then be of great value in the solution of these tasks. 

The rules for basic numerical manipulations for engineering and scientific applications are given by the IEEE 754 floating point standard~\cite{ieeefloat}. In modern day processors, alongside the arithmetic logic units which operate on integers, there exist floating-point units (FPUs) that operate on floating point numbers with some given precision~\cite{intelArch}. Whilst addition and multiplication operations are quite optimized in hardware, floating-point division, which typically employs a Newton-Raphson iteration\cite{hardware_div}, performs three to four times slower in latency and throughput.\footnote{Latency - Number of clock cycles it takes to complete the instruction. Throughput - Number of clock cycles required for issue ports to be ready to receive the same instruction again.}  In the case of matrices, inversion is typically performed in software (LAPACK, BLAS). Most algorithms must continuously communicate between the CPU and memory, eventually running into the von Neumann bottleneck. This is a prime concern for software matrix inversion. In fact, methods to minimize communication with memory to solve linear systems is an active area of research.~\cite{matrixinv} In recent years, real-time matrix inversion has been implemented in hardware by FPGAs for wireless coding. \cite{mimoinv} Even a quantum algorithm to speed up the solution of linear systems has been proposed \cite{quantumlin}. However, these hardware solutions are typically constrained to very small systems or cryogenic temperatures, respectively. 

In this paper, we avoid many of these issues by introducing a fundamentally different approach to scalar (and matrix) inversion. We propose solving the inversions in hardware by employing DMMs 
implemented with self-organizable logic circuits (SOLCs), namely circuits made of a collection of SOLGs. In this way, the need for a complicated approach to inversion is greatly simplified since the `algorithm' implemented is essentially a Boolean problem, and the computational time is effectively reduced to real time when done in hardware. Our novel, and non-trivial contribution here is extending the factorization solution in \cite{dmm} to a circuit which inverts fixed point scalars. We demonstrate that this can be done efficiently at any desired precision by proving that requiring $n$ bits of precision in the result demands extending the factorization circuit by at most $n$ more bits. To the authors' knowledge, this is the first explicit construction of a fixed point inversion, with an arbitrary specified precision, into an exact factorization between integers. We also provide a roadmap from our scalar inversion method toward full matrix inversion and the solution of linear systems, the full implementation of which is left for future work.

This paper is organized as follows: In Sec. 2 we briefly outline the concept of DMMs and SOLCs. In Sec. 3 we describe in detail how the scalar inversion problem can be solved so that the reader can follow all its conceptual and practical steps without being bogged down by details. In this section we also simulate the resulting circuit for several cases, and discuss the scalability of our approach. In Sec. 4 we describe how matrix inversion can be done by repeated application of our solution for scalars. Finally, in Sec. 5 we report our conclusions. 

\section{Digital Memcomputing Machines and Self-Organizing Logic Circuits}
The method we employ for numerical inversion is an extension of the factorization solution done by Traversa and Di Ventra in Ref.~\cite{dmm}. We build on their method to perform scalar inversion, and by generalization, find the inverse of a matrix. 

The SOLCs we utilize here are a practical realization of DMMs, which themselves are a subclass of UMMs, and thus take advantage of \emph{information overhead}, namely the information embedded in the topology of the network rather than in any memory unit, and the \emph{intrinsic parallelism} of these machines, which refers to the fact that the transition function of the machine acts simultaneously on the collective state of the system~\cite{umm}. Digital memcomputing machines can be formally defined in much the same way as Turing machines, as the following eight-tuple:~\cite{dmm}
\begin{equation}
\text{DMM}= (\mathbb{Z}_2, \Delta, \mathcal{P}, S, \Sigma, p_0, s_0, F).
\end{equation}
\noindent
Here $\Delta$ is a set of transition functions,
\begin{equation}
\delta_{\alpha}: \left(\mathbb{Z}_2^{m_\alpha} \setminus F\right) \times \mathcal{P} \to \left(\mathbb{Z}_2^{m_\alpha'} \setminus \mathcal{P}^2\right) \times S,
\end{equation}
where $S$ is a set of indices $\alpha$, $\mathcal{P}$ is a set of arrays of pointers $p_\alpha$ that select memprocessors called by $\delta_\alpha$, $\Sigma$ is a set of initial states, $p_0$ an initial array of pointers, $s_0$ the initial index, and $F$, a set of final states.

From the above abstract definition we take two concrete points relevant to our discussion here. The first point is the fact that $\delta_\alpha$, the transition functions, act on the {\it collective states} of the memprocessors, which gives rise to the intrinsic parallelism of the DMM. They also map finite sets of states to finite sets and hence describe a digital system. Secondly, we can explicitly see the novelty of DMMs, and how they differ from Turing machines. The way the DMM works is by first encoding a given problem into the topology of the network of memprocessors, which specifies the structure of the set $\mathcal{P}$, which in turn facilitates a non-trivial communication between states by the transition functions. Thus, the DMM works by leveraging the structure of a given problem as reflected in the topology of the machine. This gives rise to the information overhead, and not in any additional software or instructions. This makes the DMM a special purpose machine designed to solve a specific problem in an efficient way. Compare this to a Turing machine, which is a more general computational paradigm but lacks the ability to specialize its processing to a given task, without some external instruction set.

All this is realized in practice by SOLCs~\cite{dmm}. In these circuits, the computation is performed by first constructing the ``forward problem'' with standard logic gates, namely the Boolean circuit that would solve such a problem, if the result were known. This specifies the topology of the DMM. To be more precise, if we let $f: \mathbb{Z}_{2}^n \to \mathbb{Z}_2^m$ be a system of Boolean functions where $\mathbb{Z}_2 = \left\{0,1\right\}$, then we ask to find a solution $ {\bf y} \in \mathbb{Z}_2^n$ of $f({\bf y}) = {\bf b}$. If numerical inversion is such problem and we replace the gates of the Boolean functions with those that self-organize, the SOLC so constructed would find $f^{-1}({\bf b})$, its solution.
 
There is a fundamental difference between standard networks of logic gates and ones that self-organize, in that the latter ones allow for any combination of input/output satisfiable in the Boolean sense, i.e., {\it non contradictory}. One can, in principle, provide these logic circuits with combinations that are not satisfiable. In this case, the system will not settle to any fixed (equilibrium) point. For the formal analysis of these dynamical systems and convergence properties, we refer the reader to the details in Ref.~\cite{dmm}, and Ref.~\cite{no-chaos} where it was demonstrated that 
chaotic behavior cannot emerge in DMMs with solutions. We stress that this is not the same as a ``reversible logic gate'' that requires the same number of input and output literals and an invertible transition function~\cite{shende2003synthesis}. 
 
From a physical point of view, the entire network of gates acts as a continuous dynamical system and self-organizes into a global attractor, the existence of which was demonstrated in \cite{dmm}. This network of logic gates exploits the spatial non-locality of Kirchhoff's laws and the adaptability afforded by the time non-locality to find the solution of a given problem in one computational step if implemented 
in hardware. The fundamental point is that although the computation occurs in an analog (continuous) sense, SOLCs represent a scalable {\it digital} system since both the input and output are written and read digitally, namely require a {\it finite precision}.

\section{Detailed Analysis of Scalar Inversion}

In this section, we outline a detailed construction of a system of self-organizing logic gates which solves a scalar inversion problem consisting of fixed precision binary numbers. In terms of circuit topology, scalar inversion is almost bitwise the same as scalar multiplication, and thus our inversion circuit looks similar to the factorization circuit it is based on. However, in the inversion case a few more complications emerge. Namely, how does one map a general scalar inversion problem onto a logical circuit, and hence an exactly satisfiable problem? 

This problem can be solved by constructing an ``embedding'' into a higher dimensional space where the arithmetic is between natural numbers, and hence exact. We do this by  padding the problem with an extra register of \emph{satisfiability} bits. Also, one must be especially careful with the interpretation of the input and output since the bitwise input into the circuit does not map transparently into the numerical values of input and output. 

To clarify all this, let us consider the problem of inverting a scalar, say $a\times b = c$. We are given $a$ and $c$ with fixed-point exponents and $n$-bit mantissas where we normalize the leading bits of $a$ and $c$ to be 1, so we are guaranteed that there is always a solution.  We forgo the discussion of the sign, as the sign bit of the product is simply an XOR between the sign bits of the constituents. The task before us is to perform an analysis of the forward problem, making sure that we translate the problem into a satisfiability (SAT) problem, since this is essentially what the SOLC solves. 

In its most general form the problem is:
\begin{multline}
2^{m_a}(0.1 a_{n-2}\cdots a_{0})\times 2^{m_b} (0.b_{n-1}\cdots b_{0})  =  \\
=2^{m_c} (0.1 c_{n-2} \cdots c_{0}), 
\end{multline}
where $a_i$, $b_i$ and $c_i$ in the above equation are either 0 or 1. 

We immediately see that the following relationship between the exponents holds: $m_a + m_b = m_c$. By setting the unknown exponent value to be $m_b = m_c - m_a$, we are left with only the relationship with the mantissas. 

What remains is a detailed analysis of the scalar inversion of the binary mantissas which we now consider independently of the exponent, essentially treating the arithmetic between natural numbers -- where each mantissa can be reinterpreted explicitly as an integer. For example, the $n$ bits of $a$ would be reinterpreted now as, 
\begin{equation}
a = a_{n-1}2^{n-1} + a_{n-2}2^{n-2}+\cdots+ a_02^0
\end{equation}
The same procedure would be performed with $b$ and $c$. 

It is obvious that to satisfy the consistency of the arithmetic, the size of the mantissa of $c$ needs to be equal to the sum of the number of bits of $a$ and $b$. To that effect, we must add $n$ bits -- which we refer to as \emph{consistency} bits -- to $c$ which we set to zero. 

There now remain two issues. Since this is arithmetic between natural numbers, it must be exact, and under the current constraints it is not always the case (take $a=3$ and $c=1$, in digital representation, for example). To address this issue, we pad both $b$ and $c$ with what we call \emph{satisfiability} or SAT bits, labeled $b_f$ and $c_f$ respectively, which are not specified and are allowed to float. This is to give the problem enough freedom to be exactly satisfiable in the Boolean sense (and hence solvable by a SOLC). Below we address the issue of how many such bits one needs to ensure exact satisfiability. 

 \begin{figure}[!t]
  \centering
  \includegraphics[width=0.4\textwidth]{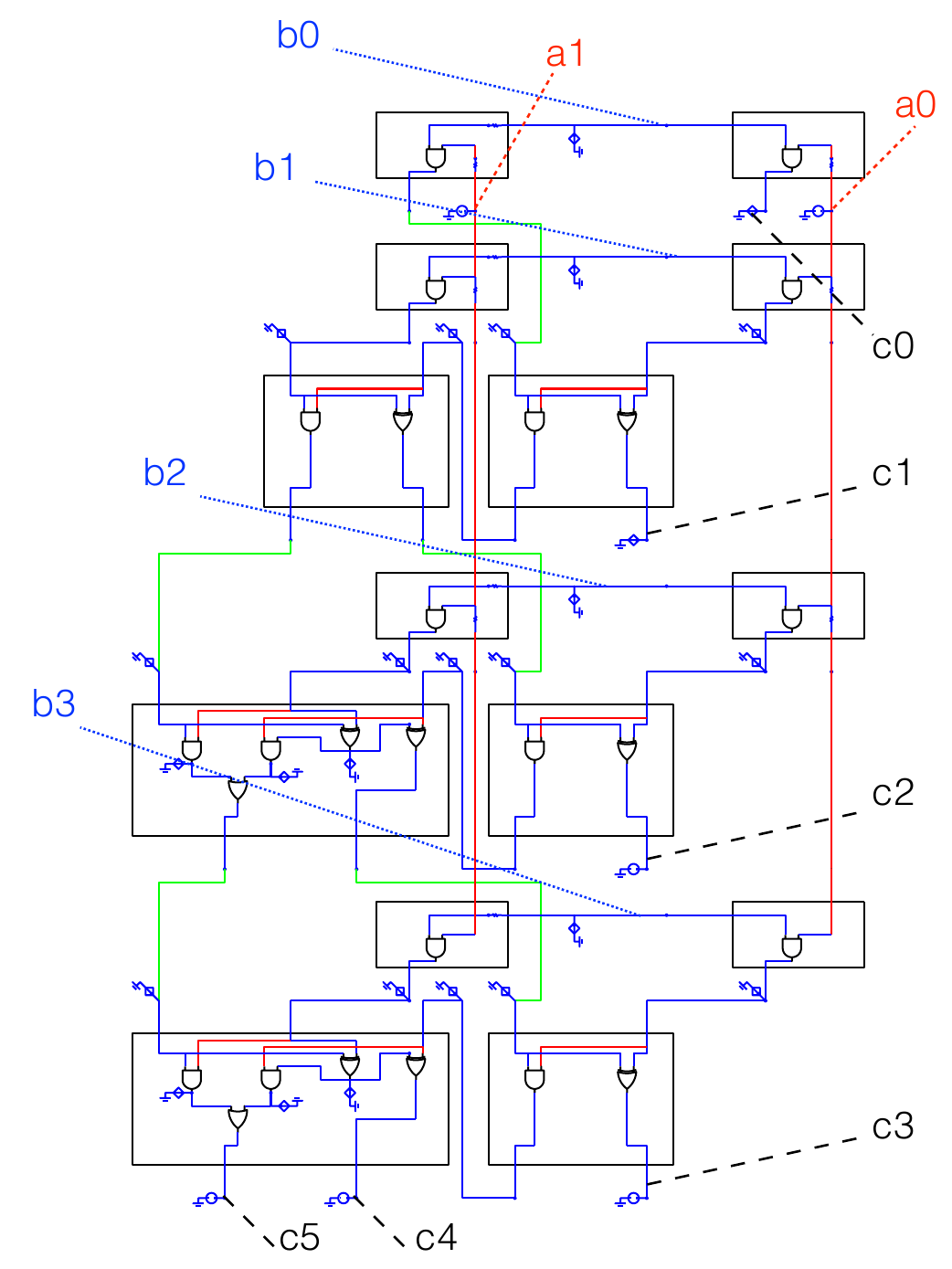}
  \caption{A schematic of a 2-bit inversion circuit. The nodes represent voltages interpreted as the bits of the input/output. The nodes are connected to a series of 2- and 3-bit adders which form the product $a\times b$ connected to the resulting bits of $c$.}

    \label{figurex}
\end{figure}

There is finally the issue of the accuracy of our answer $b$ in bits. In order to control precision, we pad $a$ with an $n_a$ number of zeros or \emph{enhanced precision} bits, which we actually show do not change the solution. 

We now show all these steps explicitly. In a more compact representation, and keeping in mind the above definitions, we write our original problem as
\begin{equation}
\underset{n}{[{\bf a}]}\underset{n_a}{[{\bf 0}]} \times \underset{n}{[{\bf b}]}\underset{n_b}{[{\bf b_f}]} = \underset{n}{[{\bf c}]}\underset{n}{[{\bf 0}]}\underset{n_a + n_b}{[ {\bf c_f}]}
\end{equation}
Here $a$, $b$ and $c$ represent the bits of the mantissas, $b_f$ and $c_f$ represent the floating bits, and $n_a$ and $n_b$ represent the size of the accuracy register and floating bit register, respectively. 
We are essentially constructing an ``embedding'' of the problem to one in a higher dimensional space where the arithmetic can indeed be satisfied exactly. The goal is to then `project' back, or truncate, onto the original space to obtain $b$ with the desired accuracy. 

The embedding is given by the following injective map where $\tilde{a}$ and $\tilde{c}$ represent the precision and consistency bits of zeroes,
\begin{align}
&\hat{a} = a2^{n_a} + \tilde{a}\nonumber \\
&\hat{b} = b2^{n_b} + b_f\\
&\hat{c} = c 2^{n_a + n_b + n} + \tilde{c}2^{n_a + n_b } + c_f\nonumber
\end{align}
Since $\tilde{a}$ and $\tilde{c}$ are identically zero, the inversion we solve, $\hat{a} \times \hat{b} = \hat{c}$ is written,
\begin{align}
a2^{n_a}(b2^{n_b} + b_f) &= c2^{n_a + n_b + n} + c_f\nonumber\\
ab2^{n_a + n_b} + ab_f2^{n_a} &= c2^{n + n_a + n_b} + c_f \label{eq5}
\end{align}
We know that the number of bits of the left in Eq.~(\ref{eq5}) must be equal to the number of bits on the right. The problem then becomes:
\begin{equation}
\overbrace{\underbracket{ \left[\blacksquare \blacksquare \right] }_{n} \underbracket{\left[\blacksquare \blacksquare \right]  }_{n} \underbracket{\left[{\bf 0}\right]  }_{n_b} \underbracket{\left[{\bf 0} \right] }_{n_a}}^{ab2^{n_a + n_b}} +
\overbrace{\underbracket{ \left[\blacksquare \blacksquare \right] }_{n} \underbracket{\left[\blacksquare \blacksquare \right]  }_{n_b} \underbracket{\left[{\bf 0}\right]  }_{n_a} }^{ab_f2^{n_a}} = 
\overbrace{\underbracket{ \left[{\bf c } \right] }_{n} \underbracket{\left[{\bf 0 } \right]  }_{n} \underbracket{\left[{\bf 0}\right]  }_{n_b} \underbracket{\left[{\bf 0} \right] }_{n_a}}^{c2^{n + n_a + n_b}}+
\overbrace{\underbracket{ \left[\blacksquare \blacksquare \right] }_{n_b} \underbracket{\left[\blacksquare \blacksquare \right]  }_{n_a} }^{c_f}\label{eq6}
\end{equation}
The black boxes represent generic non-zero bits. Here, we can see that the last $n_a$ bits of the 2nd term on the rhs of  Eq.~(\ref{eq6}) are required to be zero. Therefore, additional accuracy padding on $a$ gives us no more significant digits in our solution of $b$ and the problem reduces to:
\begin{equation}
\label{eq7}
\overbrace{\underbracket{ \left[\blacksquare \blacksquare \right] }_{n} \underbracket{\left[\blacksquare \blacksquare \right]  }_{n} \underbracket{\left[{\bf 0}\right]  }_{n_b}}^{ab2^{n_b}} +
\overbrace{\underbracket{ \left[\blacksquare \blacksquare \right] }_{n} \underbracket{\left[\blacksquare \blacksquare \right]  }_{n_b} }^{ab_f} = 
\overbrace{\underbracket{ \left[{\bf c } \right] }_{n} \underbracket{\left[{\bf 0 } \right]  }_{n} \underbracket{\left[\blacksquare \blacksquare \right]  }_{n_b} }^{c2^{n + n_b} + c_f}
\end{equation}
We are now ready to organize the main result of this section in the following theorem.
\begin{theorem} Given a scalar inversion problem with a mantissa of size $n$, the number of floating bits, $n_b$, necessary to invert the input is at most $n$.\end{theorem}
\begin{proof}
Our original problem has now been cast as one between integers $a\times \hat{b} = \hat{c}$, seen in Eq. \ref{eq7}. Assuming $a \neq 0$ (which, in our method, it is not by construction) by Euclidean division we are guaranteed that $c2^{n_b + n}= aq + r$ where $q,r \in \mathbb{Z}$ and $0 \leq r < |a|$. 
\begin{align}
a\hat{b} &= \hat{c} \nonumber \\
a(b2^{n_b} + b_f) &= c2^{n_b +n} + c_f\\
a (\underbrace{b2^{n_b} + b_f}_{q}) \underbrace{- c_f}_{r}&= aq + r\nonumber
\end{align}
Since $r < |a|$, and $a$ has $n$ bit length, one would need \emph{at most} $n_b = n$ bits to represent the floating bits, and our final output is provided by the first $n$ bits of $q$. 
\end{proof}
We conclude that to recover $n$ bits of precision in the numerical inverse, one would need to extend the register of $b$ (and for consistency, also $c$) by $n$ bits. By doing so, the scalar inversion problem becomes a problem in exact bitwise arithmetic, and thus a circuit which can be exactly satisfied and solved by a SOLC.

 \begin{figure}[!t]
  \centering
    \includegraphics[width=0.45\textwidth]{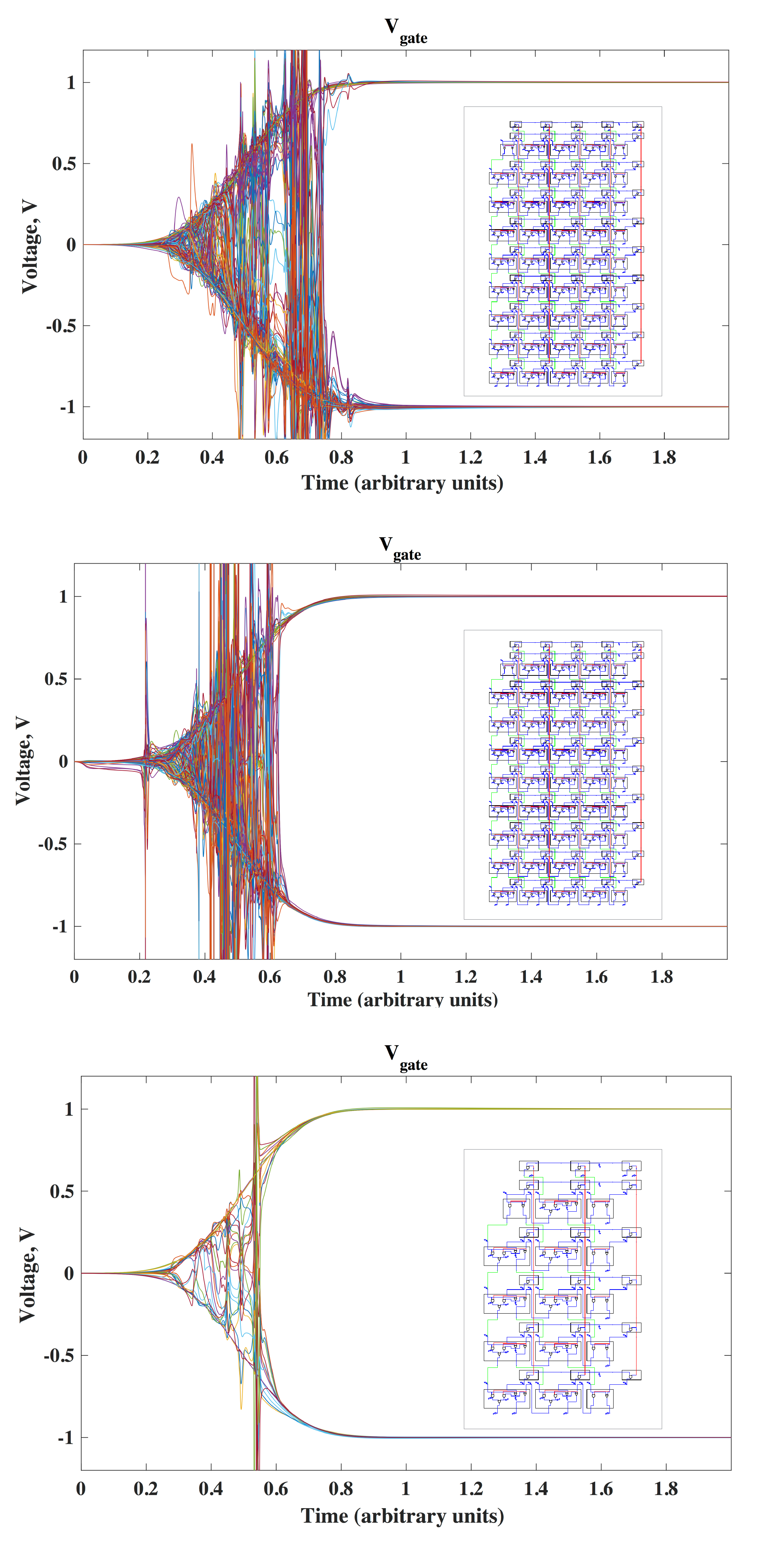}
      \caption{A plot of the voltage of all logic gate literals in several simulated systems with the corresponding circuit displayed as inset. The vertical axis represents voltages at the nodes with 1V representing the Boolean 1, and -1V the Boolean 0. The systems are found to be converged after some simulation time (arbitrary units). Plotted configurations are, from the top, inverting $a=5$ in a 5-bit circuit, inverting $a=3$ in a 5-bit circuit, and inverting $a=3$ in a 3-bit circuit.}
    \label{figurexx}
    \end{figure}

{\it Numerical Simulations-}
We constructed the solution of numerical inversion by extending the factorization solution found in \cite{dmm}. The modified circuit contains the extended registers of satisfiability bits and freely floating nodes attached to voltage-controlled differential current generators representing them. We performed simulations on up to 5-bit examples with 
different initial conditions using the Falcon$^\copyright$ simulator\footnote{Falcon$^\copyright$ simulator is an in-house circuit simulator developed by {F. L. Traversa} optimized to simulate self-organizing circuits. Some examples of other simulations can be found, e.g., in \cite{dmm,topo,13_amoeba,Traversa_TCAD,Traversa_AEU}. In all cases, the resulting $n$ bits of the inverse always matched the corresponding $n$ bits of the exact answer.}.

A simplified circuit of a 2-bit inversion is shown in figure \ref{figurex} with the internal logic gates inside the 2-bit and 3-bit adders shown for clarity. In figure \ref{figurexx} we plot the simulated voltages across several cases of scalar inversion. The topmost simulation is of a 5-bit circuit inverting $a=10_{10} = 01010_2$ as a function of time, which converges to the logical 1 (voltage $V_{gate}=1 V$) and logical zero ($V_{gate}=-1 V$) once the inverted solution is found. The solution is $b = 0.1_{10}= 0.00011001001100..._2$. While expressible in base 10 as a finite decimal, in binary the expression does not terminate. However, our circuit finds the correct 5 truncated digits of the exact solution $b \approx 0.00011_2$ in binary representation. 

A few more examples are plotted for comparison. The simulation in the middle finds the inverse of $a=3$ in a 5-bit circuit. Finally the bottom-most simulation finds the inverse of $a=3$ in a 3 bit circuit.

{\it Scaling}- The multiplication operation on two $n$ bit numbers involves $n^2$ AND gates and $n$ additions. In our case, if we fix the number of bits of the input and increase the precision $p$ of the inverse, the number of logic gates will scale as $O\left((n+p)^2\right)$. The resulting circuit scales quadratically in the number of input bits with fixed precision, and also quadratically in the precision, while fixing the number of bits of the input. Since the inversion is essentially an extended factorization, the equilibria will be reached exponentially fast and the total energy will scale polynomially in the amount of time to reach the equilibria as discussed in \cite{dmm}. 

 \begin{figure}[!h]
  \centering
    \includegraphics[width=0.5\textwidth]{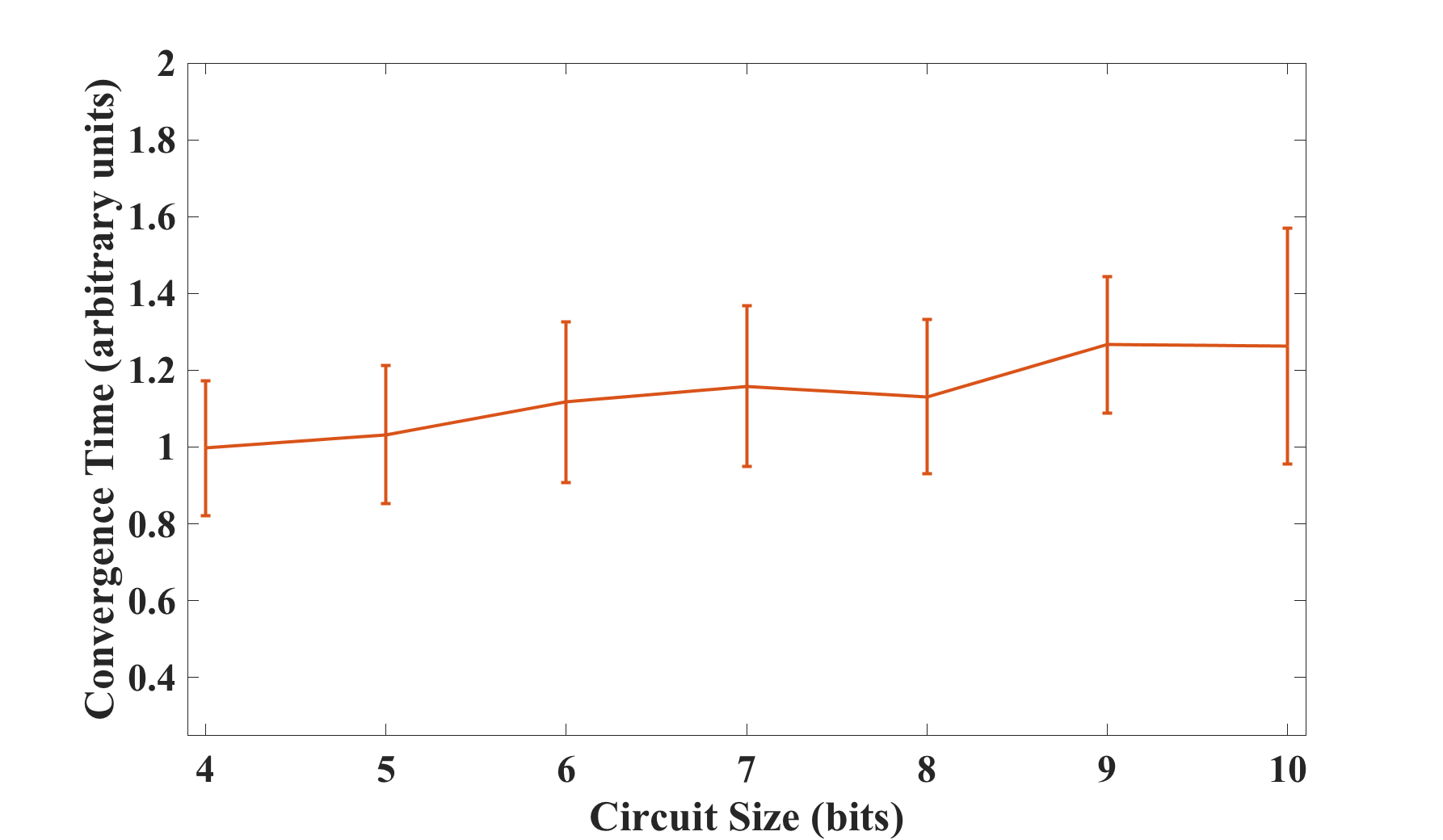}
      \caption{The average SOLC machine convergence times (proportional to the number of steps required to find the solution) across 10 simulations as a function of the size of the problem in bits. This plot shows the same scalar $(a=2)$ being inverted with a variable circuit size. Since each simulation starts with random initial conditions, error bars around the mean are shown to give a sense of the resulting dispersion.}
    \label{figure_convg}
    \end{figure}

It is also worth noting that our inversion circuit performs the search of the equilibrium point collectively, by 
employing instantons~\cite{topo}. This makes convergence times, if implemented in hardware, remain relatively independent from the system size since 
instantons are non-local objects that span the {\it entire} circuit irrespective of its size~\cite{topo}. A demonstration of this point is shown in Fig. \ref{figure_convg}, where we plot the machine time (namely the number of steps required to find the solution) versus the number of bits for a given inversion problem. We say the problem has converged when the maximum distance between all the voltages at the logic gates and true logical voltages (-1 and 1) is less than $\epsilon = 0.01$. More precisely, the simulation has reached the convergence time, $t_c$, as soon as the following inequality is satisfied:
\begin{equation}
C(t_c) = \max_i \left\{\min_{v_\text{logical} \in \{1, -1\}} |v_i(t_c) - v_{\text{logical}}| \right\} \leq \epsilon \label{Ct}
\end{equation}
Here $v_{i}(t)$ represents the voltage at the $i$-th logic gate terminal in the circuit as a function of time. Taken first is the minimum over the logical voltages 1 and -1 which represent the Boolean 1 and 0 respectively. For clarity, we plot the function $C(t)$ in Fig. \ref{figure_convg_time} together with the corresponding simulation. These numerical results give further evidence for the efficiency and scalability of the approach to numerical inversion we have developed here. 

 \begin{figure}[!h]
  \centering
    \includegraphics[width=0.5\textwidth]{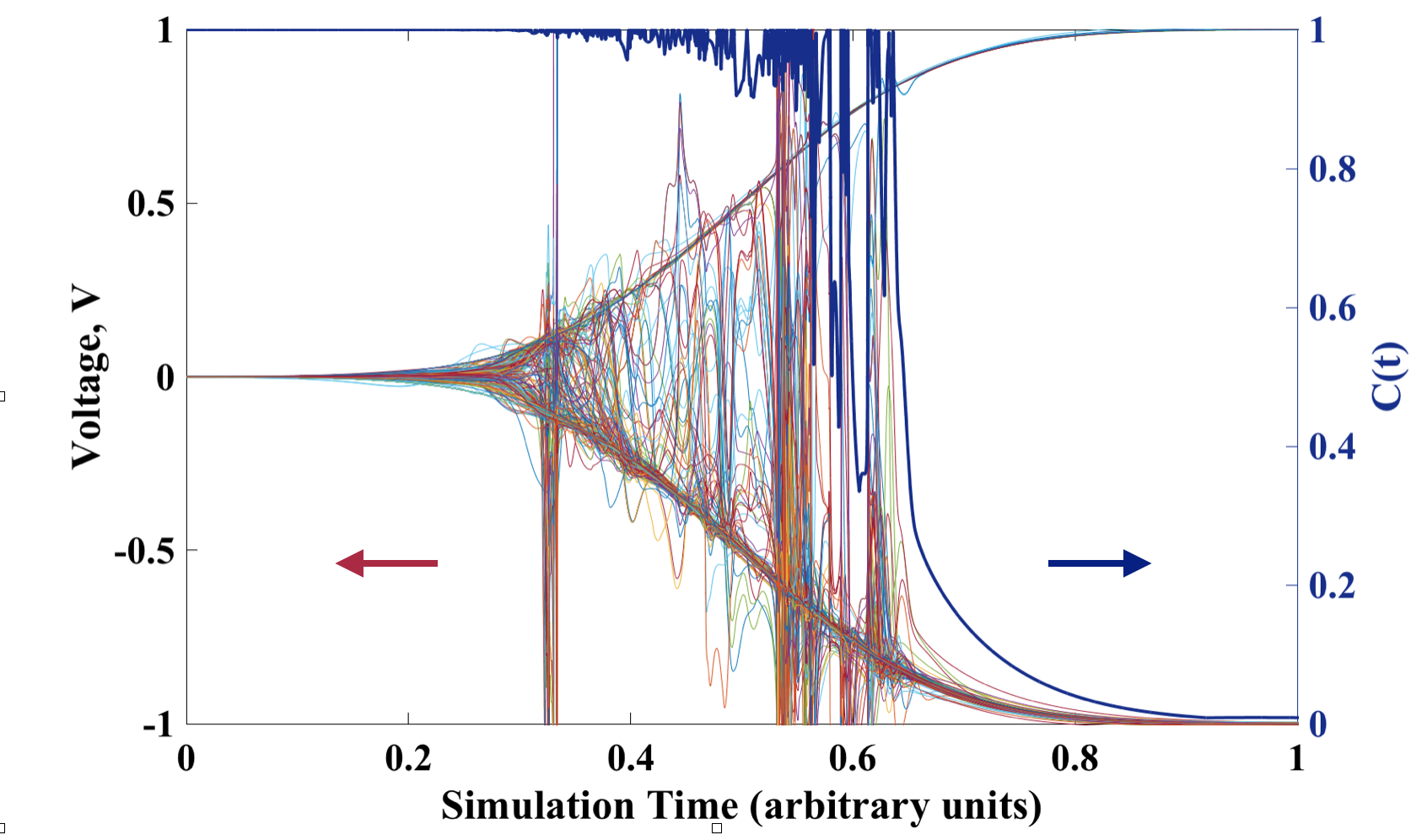}
      \caption{The function $C(t)$ [see Eq.~(\ref{Ct})] and the related SOLC simulation for the case of 4-bit inversion of $a = 2$.}
    \label{figure_convg_time}
    \end{figure}

\section{Extension to Matrix Inversion and Linear Systems}
Once we have discussed explicitly the case of scalar inversion, it is now a simple (although cumbersome) exercise to extend it to the matrix inversion case. 
However, the explicit procedure for a general matrix inversion would require too many details to fit in this paper and we will report it in a subsequent publication. We then just provide the explicit procedure given 
a $2 \times 2$ non-singular matrix $A$. 

Consider then the following matrix equation $AX = I$, where the matrix $A$ is given with $n$ bit binary entries, $X$ is the solution we seek, and $I$ is the identity matrix,
\[ \left[ \begin{array}{cc}
a_{11} & a_{12}  \\
a_{21} & a_{22}  \\
 \end{array} \right]
 \left[ \begin{array}{cc}
 x_{11} & x_{12}  \\
 x_{21} & x_{22}  \\
 \end{array} \right] =
 \left[ \begin{array}{cc}
 1 & 0  \\
 0 & 1  \\
 \end{array} \right]\,.\]

This is equivalent to solving the following two linear systems of the form $Ax = b$, one system for each column of the resulting inverse matrix.
\[ \left[ \begin{array}{cc}
a_{11} & a_{12}  \\
a_{21} & a_{22}  \\
\end{array} \right]
\left[ \begin{array}{c}
x_{11}   \\
x_{21}   \\
\end{array} \right] =
\left[ \begin{array}{c}
1   \\
0  \\
\end{array} \right]\,,\] 

\[ \left[ \begin{array}{cc}
a_{11} & a_{12}  \\
a_{21} & a_{22}  \\
\end{array} \right]
\left[ \begin{array}{c}
x_{12}   \\
x_{22}   \\
\end{array} \right] =
\left[ \begin{array}{c}
0   \\
1 \\
\end{array} \right]\,.\] 
The systems above are independent and can be solved in parallel. This gives us the following coupled equations that we must solve simultaneously with a SOLC: 

\begin{align}
a_{11}x_1 + a_{12}x_2 &= b_1\nonumber\\
a_{21}x_1 + a_{22}x_2 &= b_2\,.
\end{align}

There are 6 total arithmetical operations to be performed: 4 products and 2 sums. We approach the general inversion problem by constructing the relevant bit-wise logic circuit modules that will perform the fixed precision multiplication and addition. We perform all of the products in the manner we discussed above. The results of these four products are then summed (with signs) and set equal to $b_1$ and $b_2$ respectively. 

The signed binary addition can be performed using the 2's complement method~\cite{parhami2009computer}. First, an XOR is applied to every non-sign bit of the products ($a_{11}x_1$, $a_{12}x_2$, etc.) with the sign bit (sign($a_{11}x_1$) ..). The sign bit is then added to the result of the XOR. This makes it such that if the product was negative, it takes the 2's complement (which is flipping every non-sign bit and adding 1 to the result) or if the product is positive, then the result is not modified. These 2's complement additions are applied to the outputs of all four products which occur in our $2 \times 2$ system. After which the resulting two additions are set equal to the 2's complements of $b_1$ and $b_2$. This completes the SOLC for the linear system. The full inverse is found by applying this circuit to all columns of the given matrix.

%If the matrix is close to singular this bound cannot be satisfied and the solution will be inaccurate. Consider the example of IEEE double precision numbers which have about 16 decimal digits of accuracy. Compare this with our method of solving a linear system which is independent of the condition number of the matrix. So long as the inverse matrix exists i.e. $\kappa(A) < \infty$, then our method will find the solution or matrix inverse in the same computational time and result in the same level of precision in the output regardless of how close to singular the matrix is.
%
%Assume the entries of the input matrix $A$ were given in $n$ bits. To see how many bits would be required for a full precision answer, we simply follow the forward problem.
%
%Multiplying two $n$-bit numbers keeping full precision results in a $2n$ product. Assuming one was solving a $k \times k$ linear system, this would result in adding $k$ binary numbers each consisting of $2n$ bits. Then to solve the (forward) problem in full precision, you would need to save $2n + k$ bits. Putting this all together, one would have to provide $2n + k$ bits input to obtain $n$ digits of precision in the output.

\section{Conclusion}
We have demonstrated the power, and more importantly, the realizability, of self-organizing logic gates applied to scalar/matrix inversion and the solution of linear systems. The extensions and applications of this work are plentiful and far reaching. The method developed in the paper has direct applications and benefit to many fields (machine learning, statistics, signal processing, computer vision, engineering, optimization) that employ high-performance methods to solve linear systems. While the current work is immediately applicable to many problems, the concept can be extended to support IEEE floating point specification in the input and output for scientific computation. 

The authors envision the effectiveness of these machines to be realized in specialized hardware that is designed to interface with current computers. These DMMs built in hardware will be capable of solving efficiently many numerical problems, scalar and matrix inversion being a small subset of potential applications which include matrix decompositions, eigenvalue problems, optimization problems, training models in machine learning, etc. 

In analogy with how specialized GPUs interface with the standard CPUs of present computers to solve specific parallel problems much more efficiently than a CPU, a dedicated DMM can be constructed to work in tandem with current computers where the CPU/GPU would outsource computationally intensive problems which the DMM can solve rapidly. We thus hope our work will motivate 
experimental studies in this direction. 
\section*{Acknowledgment}
One of us (H.M.) acknowledges support from a DoD SMART scholarship. F.L.T. and M.D. acknowledge partial support from the Center for Memory and Recording Research at UCSD and LoGate Computing, Inc.

\bibliographystyle{IEEEtran}
\bibliography{IEEEabrv,IEEEexample}
%\ifCLASSOPTIONcaptionsoff
%  \newpage
%\fi
%\begin{thebibliography}{1}
%\bibitem{4}
%Lukosevicius, M., \& Jaeger, H. (2009). Survey: Reservoir computing approaches to \hskip 1em plus
%0.5em minus 0.4em\relax recurrent neural network training. Computer Science Review, 3, 127-149.
%
%\bibitem{3}
%Wolfgang Maass, Thomas Natschläger, Henry Markram, "Real-time computing without stable states: \hskip 1em plus
%0.5em minus 0.4em\relax A new framework for neural computation based on perturbations", \emph{Neural Computation} 14 (11) (2002) 2531-2560.
%
%\bibitem{fact}
%F.~Traversa and M.~Di Ventra, ``Polynomial-time Solution of Prime Factorization and NP-hard Problems with Digital Memcomputing Machines",\hskip 1em plus
%  0.5em minus 0.4em\relax \emph{arXiv:1512.05064}, 2015.
%
%\bibitem{fact2}
%F.~Traversa and M.~Di Ventra, ``Universal Memcomputing Machines",\hskip 1em plus
%  0.5em minus 0.4em\relax \emph{IEEE Trans. on Neural Networks and Learning Systems, DOI:10.1109/TNNLS.2015.2391182}, 2015.
%
%\end{thebibliography}

\end{document}